\documentclass[11pt,a4paper]{article}

\usepackage{fullpage}
\usepackage[utf8x]{inputenc}
\usepackage{amsmath, amsthm}
\usepackage{wrapfig,graphicx,amssymb,textcomp,array,amsmath}
\graphicspath{{fig/}}
\usepackage{yhmath}

\usepackage{color}
\usepackage{changepage}
\usepackage{algpseudocode}
\algtext*{EndWhile}
\algtext*{EndIf}
\usepackage{algorithm}
\usepackage{enumerate}
\usepackage{enumitem}
\usepackage{microtype}

\usepackage{microtype}
\usepackage{xspace}
\usepackage[justification=centering]{subcaption}
\usepackage{bm}

\title{Improved Bounds for Guarding Plane Graphs with Edges
}

\author{Ahmad Biniaz\thanks{Cheriton School of Computer Science, University of Waterloo. Supported by NSERC and Fields postdoctoral fellowships. \tt{ahmad.biniaz@gmail.com}}
	\and Prosenjit Bose\thanks{School of Computer Science, Carleton University. Supported by NSERC. \tt{jit@scs.carleton.ca}} 
	\and Aur\'{e}lien Ooms\thanks{D\'epartement d'Informatique, Universit\'e libre de Bruxelles (ULB), Belgium,	Supported by the Fund for Research Training in Industry and Agriculture (FRIA). \tt{aureooms@ulb.ac.be}}
	\and Sander Verdonschot\thanks{School of Computer Science, Carleton University. Partially supported by NSERC and the Carleton-Fields Postdoctoral Award. \tt{sander@cg.scs.carleton.ca}}}

\date{\today}
\newcommand{\etal}{et~al.\xspace}

\newtheorem{lemma}{Lemma}
\newtheorem{corollary}{Corollary}

\newtheorem{theorem}{Theorem}

\newtheorem*{problem*}{Problem}

\setlist[description]{leftmargin=\parindent,labelindent=\parindent}

\begin{document}

\maketitle

\begin{abstract}
  An \emph{edge guard set} of a plane graph $G$ is a subset $\Gamma$ of edges of $G$ such that each face of $G$ is incident to an endpoint of an
  edge in $\Gamma$. Such a set is said to \emph{guard} $G$.
  We improve the known upper bounds on the number of edges required to guard
  any $n$-vertex embedded planar graph $G$:
  \begin{enumerate}
  	\item We present a simple inductive proof for a theorem of Everett and Rivera-Campo (1997) that $G$ can be guarded with at most $ \frac{2n}{5}$ edges, then extend this approach with a deeper analysis to yield an improved bound of $\frac{3n}{8}$ edges for any plane graph.
  	\item We prove that there exists an edge guard set of $G$ with at most $\frac{n}{3} +
  	\frac{\alpha}{9}$ edges, where $\alpha$ is the number of quadrilateral faces
  	in $G$. This improves the previous bound of $\frac{n}{3} + \alpha$ by Bose,
  	Kirkpatrick, and Li (2003). Moreover, if there is no short path between any two quadrilateral faces in $G$, we show that $\frac{n}{3}$ edges suffice, removing the
  	dependence on $\alpha$.
  \end{enumerate}
\end{abstract}

\section{Introduction}
\label{sec:introduction}

The original Art Gallery Problem: ``How many guards are necessary, and how many are sufficient to
patrol the paintings and works of art in an art gallery with $n$
walls?" was posed by Victor Klee in 1973. Chvatal \cite{Chvatal1975} offered the first solution to the question by proving that $n/3$ guards are sufficient and sometimes necessary to guard an $n$-vertex polygon. However, since then, an active area of research was spawned, where researchers studied many different variants of the problem, by allowing different types of guards and guarding different types of objects. The field is vast and many surveys on the topic have been written (see \cite{Handbook1, Handbook2, Joebook, Shermer92}). In this paper, the variant we study is when the guards are edges and the object guarded is a plane graph.

A \emph{plane graph} is a graph that is embedded in the plane without crossing edges.
Throughout this paper, $G$ is a plane graph with $n \geq 3$ vertices and at least one edge.
The graph $G$
divides the plane into regions called the \emph{faces} of
$G$. A \emph{guard set} for $G$ is a subset $\Gamma$ of edges of $G$ such that
every face of $G$ (including the outer face) contains at least
one endpoint of an edge in $\Gamma$ on its boundary. In other words, when the endpoints of the edges of $\Gamma$ guard the faces of $G$, we say that $\Gamma$ \emph{guards} $G$.
We focus on the
problem of finding a guard set for $G$ with minimum size. To avoid some notational clutter, we omit floors and ceilings in the statements of the bounds. However, since the size is necessarily integer, all fractional bounds can be rounded down for upper bounds and rounded up for lower bounds, except in the case when the upper bound is less than 1, in which case, we round up to 1.

For maximal outerplanar graphs, O'Rourke \cite{Rourke1983} showed that $\frac{n}{4}$ edge guards are always sufficient and sometimes necessary. In his proof, both the upper bound and lower bound require that every bounded face is a triangle and the outer face is a cycle. By removing this restriction, both the upper and lower bounds jump to $\frac{n}{3}$ for arbitrary outerplanar graphs~\cite{Bose1997, Chvatal1975}. For maximal plane graphs (triangulations), Everett and Rivera-Campo~\cite{Everett1997} showed that $\frac{n}{3}$ edge guards are always sufficient and Bose~\etal~\cite{Bose1997} showed that $\frac{4n-4}{13}$ edge guards are sometimes necessary. The upper bound is derived using the four-color theorem. Note the gap between the upper and lower bounds. The lower bound is derived by constructing a triangulation where $\frac{4n-4}{13}$ triangles are \emph{isolated}. Two triangles are isolated if there is no edge joining a vertex of one triangle with a vertex of the other triangle. Since it is impossible to isolate $\frac{n}{3}$ triangles in a maximal plane graph, this would suggest that the upper bound argument may not be exploiting all of the structure present in a maximal plane graph. 

Indeed, when one studies plane graphs that are no longer restricted to be maximal, the current best upper bound is no longer $\frac{n}{3}$. Everett and Rivera-Campo~\cite{Everett1997} used the four-color theorem to prove that $\frac{2n}{5}$ edges suffice. By using a different coloring approach, Bose, Kirkpatrick and Li~\cite{Bose2003} proved that $\frac{n}{3} + \alpha$ edges are sufficient, where $\alpha$ is the number of quadrilateral faces of $G$. 
Since outerplanar graphs are planar, $\frac{n}{3}$ edges are sometimes necessary and no better lower bound is known.
Although it seems that the number of quadrilateral faces plays a key role in this problem, it is unclear which upper bound is better in the worst case: $\frac{2n}{5}$ or  $\frac{n}{3} + \alpha$, since $\alpha$ can be as high as $n - 2$. Our main contribution is an improvement on both upper bounds. 
We give a simpler proof for Everett and Rivera-Campo's upper bound of $\frac{2n}{5}$ edges. In addition, by exploiting various properties of planar graphs, we are able to strengthen the bound to $\frac{3n}{8}$ edges. We then show that, for plane graphs with $\alpha$ quadrilateral faces, $\frac{n}{3} + \frac{\alpha}{9}$ edges suffice, reducing the dependency on $\alpha$. Table~\ref{table1} summarizes the best known upper and lower bounds.

\begin{table}[htb]
	\centering
	\begin{tabular}{|c|c|c|}
		\hline 
		Graph Type & Lower Bound & Upper Bound \\ \hline
		Maximal Outerplanar & $ \frac{n}{4}$\ \cite{Rourke1983} & $\frac{n}{4}$\ \cite{Rourke1983}\\ \hline
		Outerplanar &  $\frac{n}{3}$\ \cite{Bose1997} & $\frac{n}{3}$\ \cite{Chvatal1975}\\ \hline
		Maximal Planar & $ \frac{4n-4}{13}$ \cite{Bose1997} & $\frac{n}{3}$ \cite{Everett1997} \\ \hline
		Planar &  $\frac{n}{3}$\ \cite{Bose1997}  & $\min\{\frac{n}{3} + \frac{\alpha}{9}, \frac{3n}{8} \} $\ [this paper]\\ \hline
	\end{tabular}
	\caption{The best known upper and lower bounds for various types of graphs, where $n$ is the number of vertices and $\alpha$ is the number of quadrilateral faces.}
	\label{table1}
\end{table}
%
%
%
%
\section{Iterative Guarding}
\label{sec:iterative}

We first introduce a proof strategy that iteratively builds a guard set while shrinking the graph. We use this strategy to give a simple proof of Everett and Rivera-Campo's~\cite{Everett1997} result that $\frac{2n}{5}$ edges suffice for any plane graph, before strengthening this bound to $\frac{3n}{8}$. Note that, if the graph has a single face, it can be guarded by one edge and our bounds hold so long as $n \geq 3$. In the remainder of this section, we assume that the initial graph has at least two faces.

The general strategy works as follows. Suppose we are aiming for a bound of $c n$ edges, for some constant $c > 0$. We start with an empty partial guard set $\Gamma = \emptyset$. Given a plane graph $G$, we identify a set of vertices $V'$ and edges $E'$ such that (i) the edges in $E'$ guard all faces incident to vertices in $V'$ and (ii) we have that $|E'| \leq c |V'|$. We then add all edges of $E'$ to $\Gamma$, remove all vertices in $V'$ from $G$, along with their incident edges, and repeat until $G$ has one face left; i.e. $G$ is a forest. This face has already been guarded in the penultimate step, so we return $\Gamma$ as our guard set. Since we added at most $c$ edges for every vertex we removed, its size is at most $c n$.

As a warmp-up, we use this strategy to prove the following bound for 2-degenerate graphs (an undirected graph is $k$-degenerate if every subgraph has a vertex of degree at most $k$).

\begin{theorem}
	\label{thm:2-degenerate}
	Every $2$-degenerate plane graph with $n \geq 3$ vertices can be guarded by at most $\frac{n}{3}$ edges.
\end{theorem}
\begin{proof}
	Let $G$ be a 2-degenerate plane graph with $n \geq 3$ vertices. If $G$ has one face, we guard it with a single edge and the theorem holds, so assume that $G$ has more than one face. We use the iterative strategy described above to construct a guard set $\Gamma$ for $G$ with $c = \frac{1}{3}$. Thus, all that is left to do is to describe how to find the sets $E'$ and $V'$.
	
	\begin{figure}[htb]
		\centering
		\includegraphics{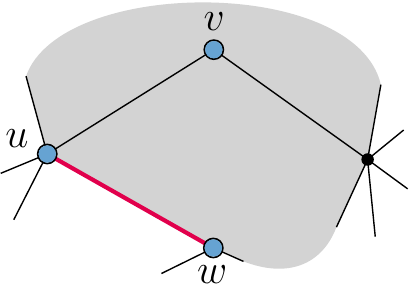}
		\caption{Edge $(u,w)$ guards both faces incident to a vertex $v$ of degree 2, allowing us to remove all three vertices.}
		\label{fig:degree-2}
	\end{figure}
	
	We consider two cases, depending on the minimum degree of $G$. If $G$ contains any vertex $v$ of degree $0$ or $1$, we let $E' = \emptyset$ and $V' = \{v\}$. While this does not technically satisfy our definition above that the edges in $E'$ guard all faces incident to vertices in $V'$, this operation is still safe, since any guard set for $G \setminus \{v\}$ is also a guard set for $G$.
	
	If $G$ does not contain any vertex of degree $0$ or $1$, the fact that it is 2-degenerate tells us that it must have a vertex $v$ of degree $2$. Let $u$ be a neighbor of $v$, and let $w \neq v$ be another neighbor of $u$ (see Figure~\ref{fig:degree-2}). Such a vertex $w$ must exist, since $G$ has minimum degree $2$. We now let $E' = \{(u, w)\}$ and $V' = \{v, u, w\}$. Since edge $(u, w)$ guards both faces incident to $v$, as well as all faces incident to $u$ and $w$, this completes the proof.
\end{proof}

This gives an alternate proof for the bound on outerplanar graphs \cite{Chvatal1975, Fisk78}, since they are 2-degenerate.

\begin{corollary}
	\label{cor:outerplanar}
	Every embedded outerplanar graph with $n \geq 3$ vertices can be guarded by at most $\frac{n}{3}$ edges.
\end{corollary}

Since a set of $\frac{n}{3}$ disjoint triangles comprises an outerplanar and 2-degenerate graph, the bounds of Theorem~\ref{thm:2-degenerate} and Corollary~\ref{cor:outerplanar} are best possible for these classes.

We use the same technique to prove the $\frac{2n}{5}$ and $\frac{3n}{8}$ bounds. Since $\frac{1}{3} < \frac{3}{8} < \frac{2}{5}$, we can use the arguments from the proof of Theorem~\ref{thm:2-degenerate} to eliminate vertices of degree $2$ or less, even if we are shooting for $c = \frac{2}{5}$ or $c = \frac{3}{8}$. Thus, we may assume for the remainder of the section that the graph has minimum degree $3$. Since planar graphs are 5-degenerate, we still need to handle vertices of degree $3$, $4$, or $5$. The following lemma gives us a little more to work with in these cases. For brevity, we denote a vertex of degree $d$ as a $d$-vertex, and one with degree at most $d$ as a $d^-$-vertex. Likewise, we denote a face with $k$ boundary edges as a $k$-face and one with at most $k$ edges as a $k^-$-face.

\begin{lemma}[Lebesgue \cite{Lebesgue1940}]
	\label{lem:lebesgue}
	In each plane graph with minimum degree $3$ there exists either a $3$-vertex
	incident to a $5^-$-face, or a $4$-vertex incident to a $3$-face, or a
	$5$-vertex incident to four $3$-faces.
\end{lemma}

\begin{figure}[htb]
	\centering
	\begin{subfigure}[b]{0.32\textwidth}
		\centering
		\includegraphics{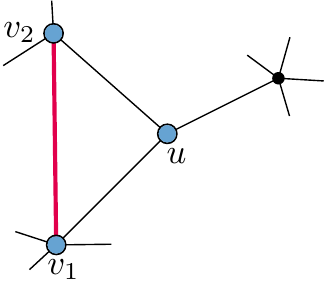}
		\caption{}
		\label{fig:remove-deg3-3}
	\end{subfigure}
	\begin{subfigure}[b]{0.32\textwidth}
		\centering
		\includegraphics{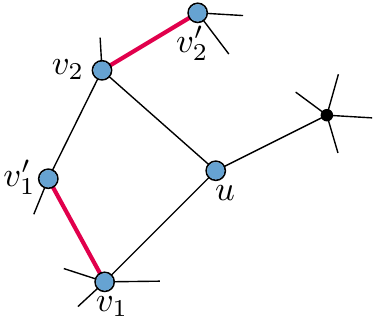}
		\caption{}
		\label{fig:remove-deg3-5}
	\end{subfigure}
	\begin{subfigure}[b]{0.32\textwidth}
		\centering
		\includegraphics{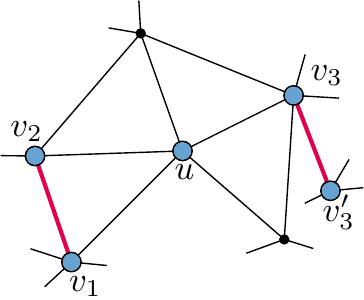}
		\caption{}
		\label{fig:remove-deg5}
	\end{subfigure}
	\caption{Guarding a vertex (a) of degree $3$ with two neighbors connected by an edge; (b) of degree $3$ with no neighbors connected by an edge; (c) of degree $4$ or $5$ and incident to a triangle.}
\end{figure}

\begin{theorem}
	\label{thm:2n5}
	Every plane graph with $n \geq 3$ vertices can be guarded by at most $\frac{2n}{5}$ edges.
\end{theorem}
\begin{proof}
	We use the iterative method with $c = \frac{2}{5}$ and, as argued above, can assume that our graph $G$ has minimum degree at least $3$.
	
	First consider the case where $G$ has a vertex $u$ of degree $3$. Any two neighbors of $u$ together are incident to all faces incident to $u$. If any two neighbors $v_1$ and $v_2$ of $u$ are connected by an edge, we let $E' = \{(v_1, v_2)\}$ and $V' = \{u, v_1, v_2\}$ (see Figure~\ref{fig:remove-deg3-3}). Otherwise, let $v_1$ and $v_2$ be any two neighbors of $u$, and let $v_1' \neq u$ be a neighbor of $v_1$ and $v_2' \notin \{u,v_1'\}$ a neighbor of $v_2$ (see Figure~\ref{fig:remove-deg3-5}). We set $E' = \{(v_1, v_1'), (v_2, v_2')\}$ and $V' = \{u, v_1, v_1', v_2, v_2'\}$.
	
	Now suppose that $G$ has minimum degree at least $4$. Then Lemma~\ref{lem:lebesgue} tells us that there must be a $5^-$-vertex $u$ incident to a triangle. Let $v_1$ and $v_2$ be the other vertices of this triangle. Edge $(v_1, v_2)$ guards three of the four or five faces incident to $u$. Let $v_3$ be a neighbor of $u$ incident to the faces not guarded by $(v_1, v_2)$, and let $v_3' \notin \{u, v_1, v_2\}$ be a neighbor of $v_3$ (see Figure~\ref{fig:remove-deg5}). We set $E' = \{(v_1, v_2), (v_3, v_3')\}$ and $V' = \{u, v_1, v_2, v_3, v_3'\}$ (see Figure~\ref{fig:remove-deg3-5}).
	
	Thus, in each case we can find $E'$ and $V'$ such that the edges of $E'$ guard all faces incident to vertices in $V'$ and $|E'| \leq \frac{2}{5}|V'|$.
\end{proof}

To improve this bound further to $\frac{3n}{8}$, we need an even stronger version of Lemma~\ref{lem:lebesgue}, inspired by Borodin~\cite{Borodin1993}. Following his terminology, an edge is incident on a face if one of its endpoints is on the face. An edge is \emph{weak} if it is incident to two triangles, \emph{semiweak} if it is incident to exactly one triangle, and \emph{strong} otherwise.

\begin{lemma}
	\label{lem:borodin}
	Every plane graph with minimum degree $3$ contains one of the following:\\[0.2em]
	\indent $(L_1)$ a weak edge joining a $3$-vertex to a $10^-$-vertex;\\[0.1em]
	\indent $(L_2')$ a weak edge joining a $4$-vertex to a $6^-$-vertex;\\[0.1em]
	\indent $(L_2'')$ a weak edge joining a $4$-vertex $u$ to a $7$-vertex $v$ such that at least one edge adjacent to $(u,v)$ around $v$ is weak;\\[0.1em]
	\indent $(L_3)$ a weak edge joining a $5$-vertex incident to at least four $3$-faces to a $6^-$-vertex;\\[0.1em]
	\indent $(L_4)$ a semiweak edge joining a $3$-vertex to an $8^-$-vertex;\\[0.1em]
	\indent $(L_5)$ a semiweak edge joining a $4$-vertex to a $5^-$-vertex;\\[0.1em]
	\indent $(L_6)$ an edge incident to a $4$-face and joining a $3$-vertex to a $5^-$-vertex;\\[0.1em]
	\indent $(L_7)$ a $5$-face incident to at least four $3$-vertices.
\end{lemma}
\begin{proof}
	Borodin~\cite{Borodin1993} proved this lemma, except with configurations $(L_2')$ and $(L_2'')$ replaced by $(L_2)$: a weak edge joining a 4-vertex to a $7^-$-vertex.
	We describe how to adapt Borodin's discharging argument to prove our stronger version. For full details, see the original paper~\cite{Borodin1993}.
	
	Initially, we assign a charge of $d - 4$ to each $d$-vertex and each $d$-face. By Euler's formula, this results in a total charge of $-8$. Then, following Borodin, we redistribute the charge as follows:
	\begin{itemize}
		\item Every face with more than $4$ sides transfers $\frac{1}{3}$ to every 3-vertex on its boundary.
		\item Every vertex transfers $\frac{1}{3}$ to each incident triangle.
		\item Each vertex $u$ transfers the following to the other endpoint $v$ of each incident edge:
		\begin{itemize}
			\item $\frac{2}{3}$ if $v$ has degree $3$ and $(u, v)$ is weak;
			\item $\frac{1}{2}$ if $v$ has degree $3$ and $(u, v)$ is semiweak;
			\item $\frac{1}{3}$ if $v$ has degree $3$ and $(u, v)$ is strong and $u$ has degree at least $6$;
			\item $\frac{1}{3}$ if $v$ has degree $4$ and $(u, v)$ is weak;
			\item $\frac{1}{6}$ if $v$ has degree $4$ and $(u, v)$ is semiweak;
			\item $\frac{1}{6}$ if $v$ has degree $5$ and $(u, v)$ is weak and $v$ is incident to four triangles.
		\end{itemize}
	\end{itemize}
	We now assume that $G$ does not contain any of the configurations $(L_1)$ through $(L_7)$, and show that this implies that every vertex and face has non-negative charge --- a contradiction. The only change from the original proof is that we cannot assume that weak edges between $4$-vertices and $7$-vertices do not exist. This only affects the part of the proof dealing with $7$-vertices, so if we can show that $7$-vertices still have non-negative charge, we are done.
	
	Consider any $7$-vertex $u$. Initially, $u$ has charge $+3$. If there is no weak edge connecting $u$ to a $4$-vertex, the original proof still applies, so suppose that $v$ is a neighboring $4$-vertex and $(u,v)$ is weak. Then $u$ transfers $\frac{1}{3}$ of its charge to $v$ and each of the two triangles incident to $(u,v)$, leaving it with $+2$ charge. Let $v_-$ and $v_+$ be the neighbors of $u$ preceding and following $v$ in clockwise order around $u$, respectively. Since $G$ does not contain configuration $(L_2'')$, neither $(u, v_-)$ nor $(u,v_+)$ is weak, so $u$ does not transfer any charge to the other faces incident to these edges. Furthermore, $v_-$ and $v_+$ must have degree at least $6$, otherwise their edge to $v$ would create configuration $(L_2')$ or $(L_5)$. Therefore they receive no charge from $u$ either.
	
	Even if the remaining faces all receive $\frac{1}{3}$ charge and the remaining vertices $\frac{1}{6}$, this would still leave $u$ with positive charge. By $(L_1)$ and $(L_4)$, no neighbor of $u$ can receive more than $\frac{1}{3}$ charge. If $u$ has another $4$-vertex $v'$ as neighbor with $(u,v')$ weak, this results in even less charge distribution, since the neighbors before and after $v'$ do not receive any charge and they cannot overlap with $v_+$ or $v_-$, since $(u, v_-)$ and $(u,v_+)$ are not weak. Finally, a $3$-vertex connected to $u$ by a strong edge would receive $\frac{1}{3}$ charge, but would prevent the adjacent faces from receiving charge. Thus, $u$ will have non-negative charge after redistribution, which completes the proof.
\end{proof}

\begin{figure}[htb]
	\captionsetup[subfigure]{labelformat=empty}
	\centering
	\begin{subfigure}[b]{0.24\textwidth}
		\centering
		\includegraphics{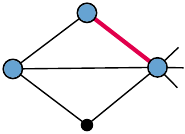}
		\caption{\bm{$(L_1)$}}
		\label{fig:L1}
	\end{subfigure}
	\begin{subfigure}[b]{0.24\textwidth}
		\centering
		\includegraphics{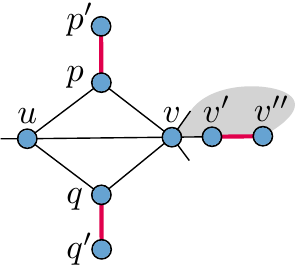}
		\caption{\bm{$(L_2')$}}
		\label{fig:L2-1}
	\end{subfigure}
	\begin{subfigure}[b]{0.24\textwidth}
		\centering
		\includegraphics{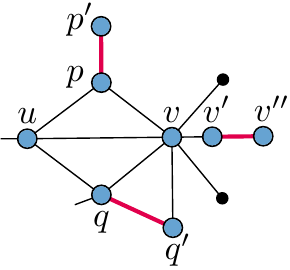}
		\caption{\bm{$(L_2'')$}}
		\label{fig:L2-2}
	\end{subfigure}
	\begin{subfigure}[b]{0.24\textwidth}
		\centering
		\includegraphics{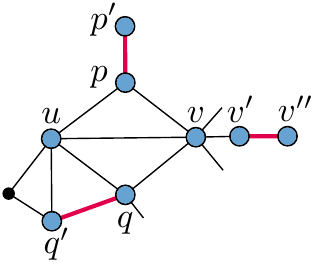}
		\caption{\bm{$(L_3)$}}
		\label{fig:L3}
	\end{subfigure}
	
	\vspace{1em}
	\begin{subfigure}[b]{0.24\textwidth}
		\centering
		\includegraphics{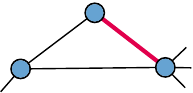}
		\caption{\bm{$(L_4)$}}
		\label{fig:L4}
	\end{subfigure}
	\begin{subfigure}[b]{0.24\textwidth}
		\centering
		\includegraphics{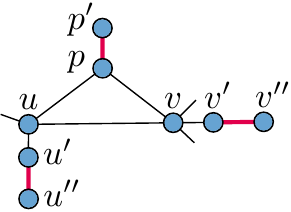}
		\caption{\bm{$(L_5)$}}
		\label{fig:L5}
	\end{subfigure}
	\begin{subfigure}[b]{0.24\textwidth}
		\centering
		\includegraphics{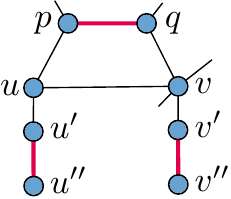}
		\caption{\bm{$(L_6)$}}
		\label{fig:L6}
	\end{subfigure}
	\begin{subfigure}[b]{0.24\textwidth}
		\centering
		\includegraphics{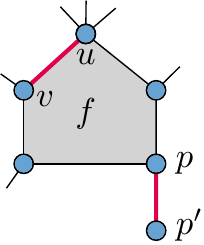}
		\caption{\bm{$(L_7)$}}
		\label{fig:L7}
	\end{subfigure}
	\caption{How to select $E'$ (thick shaded edges) and $V'$ (large shaded vertices) in each configuration of Lemma~\ref{lem:borodin}.}
	\label{fig:solve-borodin}
\end{figure}

With Lemma~\ref{lem:borodin} in hand, we can improve our bound to $\frac{3n}{8}$.

\begin{theorem}
	\label{thm:3n8}
	Every plane graph with $n \geq 3$ vertices can be guarded by at most $\frac{3n}{8}$ edges.
\end{theorem}
\begin{proof}
	As before, we use the iterative method and assume that the minimum degree of our plane graph $G$ is 3. We describe how to find $E'$ and $V'$ for each configuration of Lemma~\ref{lem:borodin} (see Figure~\ref{fig:solve-borodin}).
	
	If $G$ contains $(L_1)$ or $(L_4)$, we consider a triangle incident to the (semi) weak edge and let $E'$ be the edge of the triangle that is not incident to the $3$-vertex. Then $V'$ consists of the $3$-vertex and both endpoints of the edge in $E'$. Thus, for the remainder of the proof, we can assume that any vertex incident to a triangle has degree at least $4$.
	
	If $G$ contains $(L_2')$, let $u$ be its $4$-vertex, $v$ be its $6^-$-vertex, and $p$ and $q$ be the other vertices of the triangles incident to $(u,v)$ (we leave these definitions implicit for the remaining cases; refer to Figure~\ref{fig:solve-borodin}). We consider a neighbor $p'$ of $p$. If $p$ has an edge to $q$, we let $E' = \{(p,q)\}$ and $V' = \{u, p, q\}$, so suppose that $p' \neq q$. Since $q$ has degree at least 4, it has a neighbor $q' \neq p'$. We add $(p,p')$ and $(q,q')$ to $E'$. If this guards all faces incident to $v$, we simply set $V' = \{u, v, p, p', q, q'\}$. Otherwise, let $v'$ be a neighbor of $v$ incident to all unguarded faces (there can be at most two, since $v$ is a $6^-$-vertex). Let $v'' \neq v$ be the other neighbor of $v'$ along the boundary of one of the unguarded faces incident to $v$. We know that $v'' \notin \{p, p', q, q', u\}$, otherwise the face would already have been guarded. Thus, we can add $(v', v'')$ to $E'$ and set $V' = \{u, v, p, p', q, q', v', v''\}$.
	
	If $G$ contains $(L_2'')$, we again set $E' = \{(p, q)\}$ with $V' = \{u, p, q\}$ if edge $(p, q)$ exists. Otherwise, let $q' \neq u$ be the other neighbor of $q$ adjacent around $v$. Since $p'$ has degree at least 4, it has a neighbor $p' \neq q'$. We add $(p, p')$ and $(q, q')$ to $E'$. If all faces incident to $v$ are guarded, we set $V' = \{u, v, p, p', q, q'\}$. Otherwise, we use the same reasoning as in the previous case to find an extra edge $(v', v'')$ that guards the remaining faces around $v$.
	
	If $G$ contains $(L_3)$, let $q' \neq v$ be the other neighbor of $q$ adjacent around $u$. Since $p$ has degree at least $4$, it either has a neighbor $p' \notin \{q, q'\}$, or it is connected to both $q$ and $q'$. In the first case, we add $(p, p')$ and $(q, q')$ to $E'$ and again find a third edge $(v', v'')$ to cover the remaining faces around $v$. In the second case, $v$ must have a neighbor $v' \neq q'$ otherwise these five vertices would form a $K_5$. Then we let $E' = \{(q, q'), (v, v')\}$ and $V' = \{u, p, q, q', v, v'\}$, since $(q, q')$ guards all faces incident to both $u$ and $p$ except for the triangle $uvp$.
	
	If $G$ contains $(L_5)$, let $u' \neq p$ be the other neighbor of $u$ adjacent to $v$ around $u$. If $p$ and $u'$ are connected by an edge, let $E' = \{(p, u')\}$ and $V' = \{u, p, u'\}$. Otherwise, let $u'' \notin \{u, v\}$ be a neighbor of $u'$ and let $p' \notin \{u, v, u''\}$ be a neighbor of $p$. These neighbors exist since $u'$ and $p$ have minimum degree $3$ and $4$, respectively. We add $(u', u'')$ and $(p, p')$ to $E'$ and, if necessary, find a third edge $(v', v'')$ to cover the remaining faces around $v$ as before. Thus, we get $E' = \{(p, p'), (u', u''), (v', v'')\}$ and $V' = \{u, u', u'', p, p', v, v', v''\}$.
	
	If $G$ contains $(L_6)$, either $u$ is connected to $q$ or it has a neighbor $u' \neq q$. In the first case, we let $E' = \{(p, q)\}$ and $V' = \{u, p, q\}$. In the second case, if $u'$ is connected to any vertex $x \in \{p, q, v\}$ then that edge would cover all faces around $u$ and give us $E' = \{(u', x)\}$ and $V' = \{u, u', x\}$. Otherwise, let $u'' \neq u$ be another neighbor of $u'$. We add $(p, q)$ and $(u', u'')$ to $E'$ and again find another edge to cover the remaining faces around $v$.
	
	Finally, if $G$ contains $(L_7)$, let $f$ be the 5-face and let $u$ be a vertex of maximum degree on $f$. Let $v$ be one of $u$'s neighbors around the face and let $p$ be the vertex on $f$ not adjacent to $u$ or $v$ around $f$. If $p$ has an edge to $u$ or $v$, then that edge covers all faces around $p$, $u$, and $v$ and we are done. Otherwise, let $p' \notin \{u, v\}$ be a neighbor of $p$ not on $f$. We set $E' = \{(u, v), (p, p')\}$ and $V' = V(f) \cup \{p'\}$.
	
	Thus, in each case we can find a set $E'$ and $V'$ such that $|E'| \leq \frac{3}{8}|V'|$.
\end{proof}

\section{Guarding by Coloring}
\label{sec:coloring}

Historically, many questions about guard placement have been resolved by finding an appropriate vertex or edge coloring. Bose~\etal~\cite{Bose2003} defined a \emph{face-respecting $k$-coloring} of a plane graph $G$ as a $k$-coloring of the vertices of $G$ such that no face is monochromatic. They were particularly interested in face-respecting 2-colorings with the additional property that every face has a monochromatic edge. For brevity, we call such colorings \emph{guard colorings}. They proved the following result, which we include here as a good introduction to the general technique.

\begin{lemma}[Bose~\etal~\cite{Bose2003}]
	If a plane graph with $n \geq 3$ vertices has a guard coloring, it can be guarded by $\frac{n}{3}$ edges.
\end{lemma}
\begin{proof}
	Consider two subgraphs $G_1$ and $G_2$ of $G$, induced by the two color classes of the guard coloring. Let $M_1$ be a maximal matching in $G_1$ and $M_2$ in $G_2$. Now consider a face $f$ that has a boundary edge $e$ with both endpoints in $G_1$. Since $M_1$ is maximal, if it does not contain $e$, it must contain one of its endpoints. Otherwise, we would obtain a larger matching by adding $e$. Thus, in each case, $M_1$ guards $f$. Recall that one of the properties of a guard coloring is that every face has a monochromatic edge. This implies that $M_1 \cup M_2$ is a guard set for $G$.
	
	We now have one guard set for $G$, but we do not have a good bound on the size of this guard set. Indeed, there are examples where $M_1 \cup M_2$ contains many more than $\frac{n}{3}$ edges. To prove the lemma, we find two more guard sets for $G$ such that the total size of all three guard sets is $n$. Then the smallest of these three sets must have size at most $\frac{n}{3}$.
	
	Our second guard set starts with all edges of $M_1$, and then adds one edge incident to each vertex of $G_1$ that is not in $M_1$. Since our guard coloring has no monochromatic faces, each face has a vertex in $G_1$. Thus, this set is also a guard set for $G$. We obtain our third guard set by repeating this construction for $M_2$.
	
	The size of the first guard set is $|M_1| + |M_2|$. The other guard sets have size $|M_1| + |V(G_1)| - 2|M_1| = |V(G_1)| - |M_1|$, and $|V(G_2)| - |M_2|$, respectively. Thus, in the total size the size of the matchings cancels and we are left with $|V(G_1)| + |V(G_2)| = n$.
\end{proof}

Bose~\etal also showed that every plane graph without quadrilateral faces has a guard coloring. Thus, a natural question is whether \emph{all} plane graphs --- even those with quadrilateral faces --- have a guard coloring? In the following theorem we show that this is not the case.

\begin{figure}[htb]
	\centering
	\includegraphics{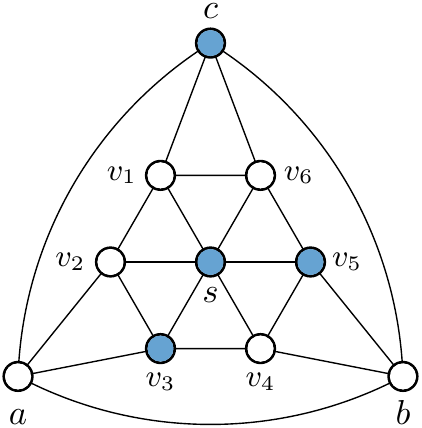}
	\caption{A plane graph without a guard coloring. The illustrated 2-coloring is forced under the assumption that $a$ and $b$ have the same color, but leaves quadrilateral $bv_5v_6c$ without a monochromatic edge.}
	\label{fig:no-guard-coloring}
\end{figure}

\begin{theorem}
	There are plane graphs that have no guard coloring.
\end{theorem}
\begin{proof}
	Consider the graph in Figure~\ref{fig:no-guard-coloring}. We need to color its vertices with two colors, say white and blue, such that every face contains (i) vertices of both colors and (ii) an edge whose endpoints have the same color. We show that such a coloring does not exist.
	
	Suppose, for a contradiction, that it does. Since the outer face is a triangle, two of its vertices must have the same color, say white. Suppose that the two vertices are $a$ and $b$; the other cases are symmetric.
	This forces $c$ to be blue, since otherwise triangle $abc$ would be monochromatic. Now either $v_1$ or $v_6$ needs to be white, otherwise triangle $cv_1v_6$ is monochromatic.
	Since the graph is symmetric, we suppose without loss of generality that $v_1$ is white.
	This forces $v_2$ to be white as well, otherwise quadrilateral $acv_1v_2$ would not have a monochromatic edge. This, in turn, forces $s$ and $v_3$ to be blue, since they are part of triangles with two white vertices. Now a sequence of such triangles forces $v_4$ to be white, $v_5$ blue, and $v_6$ white. But this leaves quadrilateral $bv_5v_6c$ without a monochromatic edge.
	Since the entire coloring was forced, this graph has no guard coloring.
\end{proof}

Note that this counter-example does not require a large guard set: $\frac{n}{5} = 2$ edges suffice. Thus, it only shows that the technique of guard colorings does not extend to all plane graphs.

Everett and Rivera-Campo~\cite{Everett1997} used a different vertex coloring to find small guard sets. We modify their approach here to give an upper bound that improves on the $\frac{n}{3} + \alpha$ bound by Bose~\etal~\cite{Bose2003}.

\begin{theorem}
	\label{thm:a/9}
	Every plane graph with $n \geq 3$ vertices and $\alpha$ quadrilateral faces can be
	guarded by at most $\frac{n}{3} + \frac{\alpha}{9}$ edges.
\end{theorem}
\begin{proof}
	We first construct a triangulation $G'$ by inserting extra diagonals in every non-triangular face of $G$, with two restrictions. First, we do not insert edges that are already in $G$. Second, for every $k$-face with $k \geq 6$ and boundary $v_1,v_2,\dots, v_k,v_1$, we first add the three edges $v_1v_3$, $v_3v_5$, and $v_5v_1$ (see Figure~\ref{fig:triangulate-6f}). By the four-color theorem~\cite{MR1025335}, we can find a proper coloring of $G'$ with a set of four colors $\{c_1, c_2, c_3, c_4\}$. Consider one such coloring, and note that it is also a proper coloring of $G$.
	
	Since each face of $G$ was triangulated in $G'$, its vertices have at least three distinct colors. Thus, if we consider any two colors, say $c_1$ and $c_2$, each face has a vertex with at least one of these two colors. In other words, each face of $G$ contains a vertex of $G_{12}$, the subgraph of $G$ induced by the vertices with color $c_1$ or $c_2$. This means we can create a guard set for $G$ by finding a set of edges whose endpoints include all vertices of $G_{12}$. We do this by finding a maximal matching $M_{12}$ of $G_{12}$, then adding one extra edge incident to each vertex of $G_{12}$ not in $M_{12}$. We call the resulting guard set $\Gamma_{12}$, and note that it contains $|\Gamma_{12}| = |V(G_{12})| - |M_{12}|$ edges, since each edge of $M_{12}$ covers two vertices in $G_{12}$. We can do this for each combination of two colors, giving us six different guard sets.
	
	\begin{figure}[htb]
		\centering
		\begin{subfigure}[b]{0.48\textwidth}
			\centering
			\includegraphics{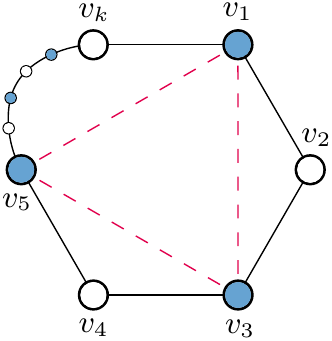}
			\caption{}
			\label{fig:triangulate-6f}
		\end{subfigure}
		\begin{subfigure}[b]{0.48\textwidth}
			\centering
			\includegraphics{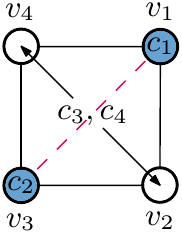}
			\caption{}
			\label{fig:triangulate-4f}
		\end{subfigure}
		\caption{A triangulation and coloring of the faces of $G$. The red dashed edges are added when triangulating (a) a face with six or more sides and (b) a quadrilateral.}
	\end{figure}
	
	Now consider the set $\Gamma_{1234} = M_{12} \cup M_{34}$. We show that this is a guard set for all non-quadrilateral faces of $G$. First, suppose that some face has an edge $e$ whose endpoints have colors $c_1$ and $c_2$. If neither endpoint of $e$ is in $M_{12}$, we can add $e$ to $M_{12}$ to obtain a larger matching. But $M_{12}$ is maximal, so it must already contain some edge incident to an endpoint of $e$. Thus, $M_{12}$ guards all faces with a $(c_1, c_2)$-edge. We claim that every non-quadrilateral face of $G$ has either a $(c_1, c_2)$-edge, or a $(c_3, c_4)$-edge and is therefore guarded by $\Gamma_{1234}$. To show this, we group colors $c_1$ and $c_2$ into one color class $c_A$ and $c_3$ and $c_4$ into $c_B$. Our claim is equivalent to saying that every non-quadrilateral face has a monochromatic edge in this two-coloring. This is clear for faces of odd length, since they cannot be properly two-colored.
	
	Let $f$ be a $k$-face with $k \geq 6$ and with boundary $v_1, \dots, v_k$ (see Figure~\ref{fig:triangulate-6f}). To avoid a monochromatic edge, the colors $c_A$ and $c_B$ must alternate along the boundary. This means that $v_1$, $v_3$, and $v_5$ get the same color. But these form a triangle in $G'$, since we started triangulating this face by inserting the edges $v_1v_3$, $v_3v_5$, and $v_5v_1$. Thus, they must have three distinct colors in the four-coloring, which means they cannot have the same color in the two-coloring. Therefore $\Gamma_{1234}$ guards all non-quadrilateral faces. An analogous argument shows that the same holds for $\Gamma_{1324} = M_{13} \cup M_{24}$ and $\Gamma_{1423} = M_{14} \cup M_{23}$.
	
	What about quadrilateral faces? Let $q$ be a quadrilateral face with boundary $v_1, v_2, v_3, v_4$ and suppose that it was triangulated by adding $v_1v_3$ (see Figure~\ref{fig:triangulate-4f}). We show that at least two of $\Gamma_{1234}$, $\Gamma_{1324}$, and $\Gamma_{1423}$ guard $q$. Suppose that $q$ is not guarded by $\Gamma_{1234}$, which means that it does not have $(c_1, c_2)$-edges, or $(c_3, c_4)$-edges. Without loss of generality, assume that $v_1$ has color $c_1$. Then the two-coloring argument and the presence of edge $v_1v_3$ force $v_3$ to have color $c_2$, while $v_2$ and $v_4$ have color $c_3$ or $c_4$. Either way, there is both a $(c_1, c_3)$- or $(c_2, c_4)$-edge and a $(c_1, c_4)$- or $(c_2, c_3)$-edge. By symmetry, this means that if one of the three does not guard $q$, the other two do. We complete $\Gamma_{1234}$ to a guard set by adding, for each quadrilateral $q$ not guarded by $\Gamma_{1234}$, one edge incident to $q$, and likewise for $\Gamma_{1324}$ and $\Gamma_{1423}$. The total size of these three guard sets is $|M_{12}| + |M_{34}| + |M_{13}| + |M_{24}| + |M_{14}| + |M_{13}| + \alpha$.
	
	We now have nine guard sets for $G$. The total number of edges in these sets is $3n + \alpha$, since each vertex occurs in three of the $G_{ij}$, and the size of the matchings cancels. Thus, the smallest of these sets has size at most $\frac{3n + \alpha}{9} = \frac{n}{3} + \frac{\alpha}{9}$.
\end{proof}

\section{Distant Quadrilaterals}
\label{sec:3-hop}

In this section, we combine both methods used previously to prove a better upper bound for plane graphs in which every pair of quadrilaterals is far apart. To make this more precise, we say that two faces $f$ and $g$ are \emph{$h$-hop apart} if every path from a vertex on the boundary of $f$ to a vertex on the boundary of $g$ contains at least $h$ edges.

\begin{theorem}
	\label{thm:3-hop}
	Every plane graph with $n \geq 3$ vertices in which every two quadrilateral faces are $3$-hop apart can be guarded by at most $\frac{n}{3}$ edges.
\end{theorem}
\begin{proof}
	We first use the iterative algorithm as described in the proof of Theorem~\ref{thm:2-degenerate} to remove any vertices of degree less than $3$. We have to be a little careful here, since removing these vertices could introduce a new quadrilateral face that is not 3-hop apart from existing quadrilaterals. To remedy this, we first mark all quadrilateral faces in the original graph. Now, if removing a vertex $v$ of degree 1 would introduce a new quadrilateral face, we instead consider its neighbor $u$ and another of $u$'s neighbors $w \neq v$ (these vertices must exist if removing $v$ would introduce a new quadrilateral). We then add $(u, w)$ to our partial guard set $\Gamma_1$ and remove all three vertices. This guarantees that all newly introduced quadrilaterals are guarded by $\Gamma_1$, since we already do the same for vertices of degree 2.
	
	If the graph was 2-degenerate, we are now done. Otherwise, this results in a graph $G$ with minimum degree at least $3$ and a partial guard set $\Gamma_1$ of size at most $\frac{n_1}{3}$, where $n_1$ is the number of vertices removed. We proceed to find a guard set $\Gamma_2$ for $G$ of size at most $\frac{n_2}{3}$, where $n_2$ is the number of vertices in $G$. The final guard set is $\Gamma_1 \cup \Gamma_2$ and has size at most $\frac{n_1}{3} + \frac{n_2}{3} = \frac{n}{3}$. Since removing vertices cannot decrease the hop distance between two faces, all marked quadrilaterals in $G$ are still 3-hop apart.
	
	We now turn to the coloring method from Theorem~\ref{thm:a/9} to find a guard set for $G$. However, we take greater care with quadrilateral faces in triangulating $G$ and constructing the matchings $M_{12}$ and $M_{34}$, to ensure that $M_{12} \cup M_{34}$ actually guards every face of $G$ instead of just the non-quadrilateral faces. Together with $\Gamma_{12}$ and $\Gamma_{34}$, this then gives us three guard sets of total size $n_2$, which means the smallest of the three has size at most $\frac{n_2}{3}$.
	
	\begin{figure}[htb]
		\centering
		\includegraphics{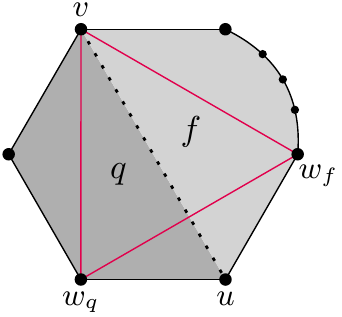}
		\caption{Triangulating the face resulting from merging quadrilateral $q$ with a neighboring face.}
		\label{fig:triangulate-q}
	\end{figure}
	
	We construct a triangulation $G'$ from $G$ as in the proof of Theorem~\ref{thm:a/9}, with one exception. If a quadrilateral $q$ does not share a boundary edge with a triangle, we merge it with one of its neighboring faces $f$ by removing the edge $(u, v)$ separating them (see Figure~\ref{fig:triangulate-q}). The result is a face with at least $7$ sides, since $f$ was not a triangle and all quadrilaterals are further apart. Let $w_f \neq v$ be the other neighbor of $u$ along the boundary of $f$, and $w_q \neq v$ the other neighbor of $u$ along the boundary of $q$. We insert edges $(v, w_f)$, $(v, w_q)$, and $(w_f, w_q)$, then triangulate the rest of the face as usual.
	
	Next, we four-color $G'$ and consider the resulting coloring of $G$. Note that the edges we removed could be monochromatic, but this is not a problem. Let $G_{12}$ and $G_{34}$ be the subgraphs of $G$ induced by all vertices with colors in $\{c_1, c_2\}$ and $\{c_3, c_4\}$, respectively. First, suppose $M_{12}$ is an arbitrary maximal matching in $G_{12}$ and $M_{34}$ in $G_{34}$. Since each face of $G$ contained a triangle in $G'$, it has vertices of at least three different colors. Therefore we still obtain guard sets $\Gamma_{12}$ and $\Gamma_{34}$ by taking the matchings and adding an edge incident to every vertex of the right colors not in the corresponding matching. Similarly, as argued in the proof of Theorem~\ref{thm:a/9}, $M_{12} \cup M_{34}$ guards all non-quadrilateral faces of $G$. We now show how to pick initial edges for $M_{12}$ and $M_{34}$ such that $M_{12} \cup M_{34}$ also guards the marked quadrilateral faces of $G$. Recall that the unmarked quadrilateral faces of $G$ are already guarded by $\Gamma_1$.
	
	Initially, $M_{12}$ and $M_{34}$ are empty. If a marked quadrilateral $q$ shares a boundary edge with a triangle $t$, then the vertices of $t$ have three distinct colors. Therefore one of the edges of $t$ must belong to $G_{12}$ or $G_{34}$, and we add this edge to the corresponding matching. If $q$ does not share an edge with a triangle, we merged it with a neighboring face by removing edge $(u, v)$. Suppose that $u$ has a color in $\{c_1, c_2\}$. Since three of its neighbors in $G$ --- $v$, $w_f$, and $w_q$ --- formed a triangle in $G'$, one of them must also have a color in $\{c_1, c_2\}$, and we add this edge to $M_{12}$. If $u$ has a color in $\{c_3, c_4\}$, we add the corresponding edge to $M_{34}$.
	
	Thus, we seed $M_{12}$ and $M_{34}$ with edges that together guard all marked quadrilateral faces of $G$. We then complete these sets to maximal matchings by greedily adding edges of $G_{12}$ and $G_{34}$, respectively. This makes $M_{12} \cup M_{34}$ a third guard set. The only thing left to argue is that none of the seed edges share an endpoint. This is guaranteed by the 3-hop distance between marked quadrilaterals in $G$; since each seed edge is incident to a marked quadrilateral, two seed edges sharing an endpoint would give a 2-hop path between two marked quadrilateral faces.
\end{proof}

\section{Conclusion}
\label{sec:conclusion}

Our main contribution lies in the development of techniques that allowed us to improve the upper bound on the number of edge guards that suffice to guard a plane graph. The role of quadrilateral faces in the size of these guard sets is intriguing. Of our bounds, one depends on the number of quadrilateral faces, while the other does not. The first bound ($\frac{n}{3} + \frac{\alpha}{9}$)  almost matches the lower bound for graphs with few quadrilateral faces, while the second bound ($\frac{3n}{8}$) is stronger for graphs with many quadrilaterals -- the two bounds balance at $\alpha = \frac{3n}{8}$ since  $\frac{n}{3} + \frac{3n}{72} = \frac{3n}{8}$. It is interesting that quadrilateral faces are the limiting factor in all techniques based on graph colorings. In contrast, our iterative technique appears to be limited by the local nature of the operation. Thus, the solution may lie in a more global approach that does not stumble over quadrilateral faces.

We leave as an open question to close the gap between the upper and lower bounds, both for maximal planar graphs and general planar graphs.

\bibliographystyle{abbrv}
\bibliography{edge-guarding}

\end{document}